\documentclass[12pt,a4paper,normalheadings,headsepline,headinclude,bibtotoc,fleqn,review]{article}
\usepackage[latin1]{inputenc}
\usepackage{amsmath}
\usepackage{booktabs}
\usepackage{array}
\usepackage{amsthm}
\usepackage{dcolumn}
\usepackage{endnotes}
\usepackage{units}
\usepackage{marvosym}
\usepackage[lines=10]{geometry}
\usepackage{longtable}
\usepackage{rotating}
\usepackage{expdlist}
\usepackage{graphicx}
\usepackage{wrapfig}
\usepackage{expdlist}
\usepackage{amssymb}
\usepackage{gloss}
\usepackage{nomencl}
\usepackage{natbib}

\newtheorem{lemma}{Lemma}
\newtheorem{prop}{Proposition}
\newtheorem{coro}{Corollary}

    \makegloss

\renewcommand{\thesection}{\arabic{section}}

\usepackage{fancyhdr} 
\begin{document}


\begin{center}\huge{An Equilibrium Model with Computationally Constrained Agents}\footnote{I am particularly indebted to Martin
Hellwig for detailed comments on an earlier draft of this paper. I
also thank Dominik Grafenhofer, Sebastian Klein, Harvey Lapan, and
Carl Christian von Weizsäcker for discussions on equilibrium
models. Finally, I received helpful questions and comments from seminar participants in Bonn. First draft July 2015.}\end{center} 

\begin{center} \emph{Wolfgang Kuhle}\\\emph{Max Planck Institute for Research on Collective Goods, Kurt-Schumacher-Str. 10, 53113 Bonn, Germany. Email:
kuhle@coll.mpg.de.}
\end{center}




\noindent\emph{\textbf{Abstract:} We study a large economy in
which firms cannot compute exact solutions to the non-linear
equations that characterize the equilibrium price at which they
can sell future output. Instead, firms use polynomial expansions
to approximate prices. The precision with which they can compute
prices is endogenous and depends on the overall level of supply.
At the same time, firms' individual supplies, and thus aggregate
supply, depend on the precision with which they approximate
prices. This interrelation between supply and price forecast
induces multiple equilibria, with inefficiently low output, in
economies that otherwise have a unique, efficient equilibrium.
Moreover, exogenous parameter changes, which would increase output
were there no computational frictions, can diminish agents'
ability to approximate future prices, and reduce output. Our model
therefore accommodates the intuition that interventions, such as
unprecedented quantitative easing, can put agents into ``uncharted territory''.}

\section{Introduction}\label{sg1}

Few people would claim that they are able to compute future
equilibrium outcomes, such as prices, with any accuracy. Despite
this, textbook models implicitly assume that, given all relevant
data, agents compute exact numeric values for future equilibrium
prices, respectively, the entire distribution of these prices if
the model involves risk. In this paper, we assume that economic
agents are computationally constrained to the use of polynomial
functions. That is, instead of being able to solve arbitrary
non-linear problems, they rely on polynomial expansions to
approximate future equilibrium outcomes. Put differently, agents
act just like economic researchers who use polynomials, such as
the Arrow-Pratt approximation, to restate complicated non-linear
problems in terms of workable polynomials.

Using a two-period model, in which firms employ polynomial
approximations to infer future selling prices for their output, we
find that multiple equilibria emerge in well-behaved economies
that would have a unique, efficient equilibrium if agents could
compute future equilibria with perfect accuracy. Moreover,
exogenous parameter changes, which would increase economic
activity were agents computationally unconstrained, can reduce
economic activity as they make it harder for agents to approximate
equilibrium prices.

Our results rely on the fact that there are levels of supply where
a polynomial approximation to the function, which relates
equilibrium supply to equilibrium price, is of good quality, and
other levels where it is of low quality. Put differently, a firm's
ability to compute equilibrium prices changes with the level of
aggregate supply. At the same time, individual supply, and thus
aggregate supply, varies with the precision with which agents can
predict prices. This interaction gives rise to two coexisting
types of equilibria. In the first, economic activity falls into
intervals where agents' polynomial approximations are of high
quality and the role of the computational friction is small. These
equilibria can coincide with the rational expectations equilibrium
(REE). In the second type, computational frictions are important
and agents find it difficult to predict prices: Aggregate supply
is (i) low and (ii) falls into an interval where agents'
approximations, to the equation describing equilibrium, are of
poor quality.

In one interpretation, we may think of a farmer who must decide in
spring how much corn he should plant. This farmer may know the
price at which corn tends to sell in years with ``normal" supply.
Moreover, he might know that small increases in aggregate supply
tend to reduce prices, i.e, that demand is locally downward
sloping. Finally, he may know that this downward slope tapers off
as supply increases. The farmer, however, is unable to calculate
all numeric values that the demand function takes over its entire
domain. If he wants to calculate those prices, which obtain once
supply differs from those levels that he is familiar with, he must
use a polynomial expansion to extrapolate the new price. In a
macroeconomic interpretation, we think of a large number of firms
that have to choose production today in anticipation of future
demand. These firms know the price at which their goods sell in
``normal times". However, if firms collectively cut production
today, it will be difficult for them to know whether future
selling prices increase, due to reduced supply, or fall, since the
layoffs, associated with production cuts, reduce
demand.\footnote{\citet{Dia82}, and \citet{Coo88}, develop models
where search frictions result in upward-sloping aggregate demand
functions. See \citet{Hel93} for a review of models with
non-monotonic demand. More generally, such effects are important
in economies where Say's law, stating that supply creates its own
demand, is of relevance. See also \citet{Cha14} for a related
model of savings and investment.} This intuition extends naturally
to economies where demand concerns a vector of goods, which may
involve substitutes and complements. In such a setting it appears
even more natural to assume that firms cannot solve for the
overall equilibrium. Instead, a firm, which produces a particular
good, may, if it is exceptionally well informed, use the economy's
Jacobian matrix to compute demand for its particular good in
terms of a first-order polynomial approximation. 

Regarding parameter changes, the uncertainty that agents face in
our model does not originate from a world with stochastically
changing parameters. Instead, agents know the magnitude of the
parameter change in advance; the difficulty is to predict its
consequences. As an illustration, we refer to two representative
comments made on the quantitative easing program: Stanley
Druckenmiller, an accomplished investor with a thirty-year track
record, commented in retrospect ``I didn't know how it was going
to end... I would have said inflation [which] would have been dead
wrong."\footnote{Speech given at the 2015 Dealbook conference.}
Similarly, taking an ex-ante perspective, Joseph Stiglitz
predicted that QE2 would likely bring interest rates down ``a
little", but that it was unclear what the risks, ranging from
economic growth to``a whole set of other potential risks that -
may result from this policy"\footnote{Interview given on the
Charlie Rose show on 3 November 2010.}, were.\footnote{See
\citet{Sor94} for case studies highlighting market participants'
difficulties to anticipate the impact that pre-announced central
bank policies have on macroeconomic equilibrium variables.
Similarly, \citet{Nie97}, p. 381, concludes his discussion on
excess demand functions: ``The difficulty is that nobody knows
what the equilibrium level is until at least the morning after the
fact."} Likewise, there appears to be no consensus among observers
how, if at all, a UK exit from the EU will affect the economies of
the UK and the remainder EU. Similar arguments apply to more
long-term problems such as the demographic
transition.\footnote{That is, we know from birth statistics that
cohorts entering the labor market will be smaller and cohorts
entering retirement will grow. At the same time, it proves
difficult to predict how such changes impact future growth paths.}
As we show, agents' inability to perform comparative statics
alters model predictions considerably: Large parameter changes,
which would unambiguously increase output if agents could compute
the model's comparative statics correctly, can reduce output. The
model therefore accommodates the intuition that interventions,
such as unprecedented quantitative easing, can put agents into
``uncharted territory'', i.e., diminish agent's ability to
forecast relevant equilibrium variables. Hence, even though
markets work inefficiently, the government's ability to improve
market outcomes is limited.

\emph{Related Literature:} Due to their bounded computational
capacity, agents work with an approximate, misspecified model.
Except for special cases, they cannot form rational expectations
in the sense of \citet{Hut37}, \citet{Gru54}, \citet{Mut61},
\citet{Bla79} and \citet{DeC79}, which are consistent with the
true model. Regarding model misspecification, our approach is thus
akin to the literature on learning, \citet{Bra82},
\citet{Mar89}, and \citet{Sar93}, where agents use a misspecified 
least squares approach to infer unknown model parameters.
\citet{Rot74} and \citet{Mcl84} model firms that experiment with
different supply functions to learn about stochastic demand. Firms
in our model are small, and thus changing individual supply has no
influence on prices. That is, in the dynamic extension of our
model, the information that firms learn over time is determined by
overall equilibrium rather than individual experimentation.

We interpret our baseline model as a simple $A^D,A^S$ setting, as
in \citet{Key36} and \citet{Sam09}, which is augmented with a
computational friction. In Section \ref{ss16}, we show that our
model may be reinterpreted as a \citet{Dia82} and \citet{Coo88}
aggregate search model, where the probability of finding a trading
partner depends on the equilibrium level of economic activity. In
this interpretation, agents' computational constraint makes it
difficult for them to compute the equilibrium probability of
finding a trading partner. Regarding government intervention,
\citet{Dia82} and \citet{Coo88} find positive multipliers, which
are due to the search friction. In the current model, where the
search friction is coupled with a computational friction,
government intervention has a non-monotonous effect on output.

\citet{Tes06}, \citet{Gin07}, \citet{Far09}, and \citet{Thu12}
argue for ``agent-based'' models in which agents follow decision
rules that do not necessarily coincide with rational
behavior.\footnote{One argument for such a departure is
computational complexity: \citet{Rub98b}, \citet{May11} and
\citet{Ack11} emphasize that agents might be constraint in their
ability to count or to compute conditional probabilities.
\citet{Nel85}, \citet{Ula08}, \citet{Hof11},
\citet{Art15}, and \citet{Kuh16} for evolutionary models where biases emerge endogenously.} 
In the current paper, agents are computationally constrained to
the use of polynomial expansions. This case is of special interest
since researchers in economics, physics, and engineering indeed
rely on first- and second-order polynomial expansions, rather than
exact solutions, to understand non-linear problems; equilibrium
comparative statics of a well behaved model, $y=f(y;b)$, are
commonly evaluated in terms of a first-order polynomial expansion 
$\Delta y\thickapprox\Delta y f_y+ \Delta b f_b$, which yields
$\frac{\Delta y}{\Delta b}\thickapprox\frac{f_b}{1-f_y}$,
respectively, $\frac{\partial y}{\partial
b}=\frac{f_b}{1-f_y}$.\footnote{Similarly, \citet{Mas95}, pp.
599-641, use first-order polynomial expansions to examine
non-linear demand in pure exchange economies with many
commodities. As mentioned earlier, \citet{Def52}, \citet{Arr71},
\citet{Pra64} use second-order expansions to study expected
utility. Likewise, the familiar first- and second-order
conditions, $f'(x_0)=0$ and $f''(x_0)<0$, for a smooth function
$f(x)$ to have local maximum at point $x_0$, stem from an
expansion
$f(x)=f(x_0)+f'(x_0)(x-x_0)+\frac{1}{2}f''(x_0)(x-x_0)^2+O^3$; see
\citet{Chi05}, pp. 250-253, or \citet{Sam47}, pp. 357-379.
Related, the stability of differential and difference equations,
\citet{Sam47}, pp. 21-121, 257-349, and 380-439, or \citet{Gal07}.
Finally, \citet{DeC79}, p. 52, points out that solving
expectations equilibria requires either that the model is assumed
to be linear, or that the model's equations have to be linearized,
i.e., rewritten in terms of a first-order polynomial, which is
what our agents do.}
Hence, we argue that the current model is methodologically
consistent in the sense that the outside researcher, i.e., the
paper's reader, will use the same method of analysis that is used
by the model's agents.


Section \ref{s1} abstracts from computational frictions and
identifies the unique rational expectations equilibrium. Section
\ref{sg2} studies equilibria that obtain with computationally
constrained agents. Section \ref{s3} considers the impact of
exogenous parameter changes. Section \ref{s5} studies the
economy's convergence to the rational expectations equilibrium in
a dynamic setting, where firms accumulate empirical knowledge.
Section \ref{s4} introduces asymmetric information. 
In Section \ref{ss16}, we suggest different interpretations of our
baseline model. Section \ref{s7} concludes.


\section{Model}\label{s1}
We study a large economy in which a mass one of firms $i\in[0,1]$
produce a homogenous good. There are two periods of time. In the
first period, each firm chooses to produce a quantity of goods
$a_i$ in anticipation of a future selling price $\hat{P}$. In the
second period, firms sell the finished products $a_i$
inelastically to consumers at a market clearing price $P$. For
simplicity, to ensure uniqueness of the REE in the economy without
computational friction, we assume that aggregate demand is twice
continuously differentiable and monotonously downward-sloping in
goods quantity $A$:
\begin{eqnarray} P=\phi(A), \quad \phi_A<0, \quad \phi_{AA}\gtrless0,\quad \phi(0)>0.\label{e1}\end{eqnarray}
Demand (\ref{e1}) represents the model's non-linearity,
respectively, the computational obstacle that agents have to
overcome. In Section \ref{ss16}, we suggest three different
interpretations of $\phi()$ by showing that it captures the
non-linearities that individual agents face when they make
forward-looking decisions in the standard workhorse models of
\citet{Dia65}, \citet{Dia82}, and the $A^D,A^S$ model of
\citet{Sam09}. That is, $\phi$ may be interpreted as (i) future
returns to savings, (ii) the probability of finding a trading
partner, or (iii) the selling price for output.

Firm $i$ chooses a production schedule $a_i^*$ to maximize
expected profits
\begin{eqnarray} a^*_i=\underset{a_i}{\arg\max}\Big\{\pi_i=a_i\hat{P}-\frac{1}{2}a_i^2\Big\},\quad a_i\geq 0,\nonumber\end{eqnarray}
where $\hat{P}$ is the firm's expectation regarding the selling
price and $\frac{1}{2}a_i^2$ is a quadratic cost function. Hence,
agent $i$ supplies
\begin{eqnarray} a_i^*=\hat{P}\nonumber\end{eqnarray}
and aggregate supply is
\begin{eqnarray}
A=\int_{[0,1]}a^*_idi=\hat{P}.\label{e33}\end{eqnarray} If agents
are computationally unconstrained, they can compute demand
(\ref{e1}) over its entire domain. That is, for each level of
aggregate supply $A$, they form rational price expectations
$\hat{P}=\phi(A)$. In turn, they combine (\ref{e1}) and
(\ref{e33}) to calculate the unique equilibrium quantity $A_0$:
\begin{eqnarray} A_0=\phi(A_0),\label{e4}\end{eqnarray}
and, using (\ref{e1}), they compute equilibrium price $P_0$:
\begin{eqnarray}P_0=\phi(A_0).\label{e41}\end{eqnarray}
Accordingly, we have
\begin{lemma}\label{p1} There exists a unique, rational expectations equilibrium $\{A_0,P_0\}\in\mathcal{R}^2_+$.
In this equilibrium agents forecast prices correctly
$\hat{P}=P_0$.\end{lemma}
\begin{proof} Market clearing (\ref{e4})-(\ref{e41}) determines equilibrium quantity $A_0>0$, which is unique since $\phi(0)>0$ and $\phi_A<0$.
Using (\ref{e1}), the equilibrium price is $P_0=\phi(A_0)=A_0>0$.
Finally, (\ref{e33}) indicates that $\hat{P}=A_0$ and thus
$\hat{P}=P_0$.\end{proof}


\subsection{Polynomial Equilibria}\label{sg2}

We now assume that firms cannot compute demand over its entire
domain. Instead, they are familiar with a point on the demand
function, $A^*,\phi(A^*)$, and the demand function's slope
$\phi_A(A^*)$ at this point. It is convenient to start with the
assumption that this point is the REE of Lemma \ref{p1}, i.e,
agents know $A_0,\phi(A_0)$, and the slope $\phi_A(A_0)$. In turn,
once supply differs from $A_0$, agents use polynomial expansions
to extrapolate demand to estimate the resulting price. Polynomial
equilibrium points will be those points where the polynomial,
which mimics true demand, intersects with supply. Put differently,
``polynomial equilibria" are those points that solve the agents'
approximate model. The REE from the previous section will be one,
but not the only, such equilibrium.


Expanding demand (\ref{e1}) around the perfect foresight
equilibrium, agents forecast the selling price (\ref{e1}) as:
\begin{eqnarray} \hat{P}=\phi(A_0)+\phi_A(A_0)\Delta A, \quad \Delta A=A-A_0. \label{hat}\end{eqnarray}
Equation (\ref{hat}) reflects that agents cannot numerically
compute the true price $P=\phi(A_0+\Delta A)$ at which a supply
$A=A_0+\Delta A$ sells. The reliability of estimate (\ref{hat})
decreases the more aggregate supply $A$ differs from
$A_0$.\footnote{Note that the model's coefficients throughout can
be chosen such that the equilibrium deviation $\Delta A$ is
arbitrarily small, respectively, such that the approximation
(\ref{hat}) is of arbitrarily good quality.} We assume that agents
choose output $a_i$ to maximize:
\begin{eqnarray} a^*_i=\underset{a_i}{\arg\max}\Big\{\pi=a_i\hat{P}-\frac{1}{2}a_i^2-a_i\tau\Delta A^2\Big\},\quad \tau \geq 0.\label{e2}\end{eqnarray}
The profit criterion (\ref{e2}) allows for two interpretations. 
In the first, $a_i\hat{P}-\frac{1}{2}a_i^2$ is the firm's profit
given the price estimate $\hat{P}$, and $-a_i\tau\Delta A^2$
reflects that firms, knowing their estimate is based on a
first-order expansion, which neglects second-order terms, discount
$\tau>0$ the estimated revenue. 
In a second interpretation, which we elaborate on in Proposition
\ref{p3} of Appendix \ref{A2}, $-\tau$ represents the demand
function's second derivative $\phi_{AA}$. In this case agents do
not discount their price estimate, and rely on a second-order
Taylor-series expansions to estimate the selling
price.\footnote{That is, if agents knew demand's second
derivative, their price estimate (\ref{hat}) would write
$\hat{P}=\phi(A_0)+\phi_A\Delta A+\frac{1}{2}\phi_{AA}\Delta A^2$.
Substituting this into (\ref{e2}), and setting the discount rate
$\tau=0$, yields
$a^*_i=\underset{a_i}{\arg\max}\Big\{a_i(\phi(A_0)+\phi_A\Delta
A+\frac{1}{2}\phi_{AA}\Delta A^2)-\frac{1}{2}a_i^2\Big\}$.
Comparison indicates that the new profit criterion is equivalent
to the old, (\ref{e2}), except for the second-order derivative
$\frac{1}{2}\phi_{AA}$ taking the place of the discount rate
$-\tau$.}

From (\ref{e2}), we obtain individual and aggregate supply:
\begin{eqnarray} a_i^*=\hat{P}-\tau\Delta A^2,\quad A=\int_{i\in[0,1]}a_idi=\hat{P}-\tau\Delta A^2.\label{e3}\end{eqnarray}
Combining supply (\ref{e3}) and estimated demand (\ref{hat}), the
equilibrium quantity $A$, where supply intersects with the demand
estimate, is the solution to:
\begin{eqnarray} A=\phi(A_0)+\phi_A(A_0)\Delta A-\tau\Delta A^2. \label{e5} \end{eqnarray}
For convenience, we identify equilibria $j=0,1,2...$, in terms of
their distance $\Delta A_j=A_j-A_0$ to the rational expectations
equilibrium $A_0$. That is, $\Delta A_j=0$ corresponds to the REE.
Using the fact that $A_0=\phi(A_0)$, we rewrite (\ref{e5}) as:
\begin{eqnarray} \tau\Delta A^2+(1-\phi_A)\Delta A=0, \nonumber \end{eqnarray}
and note:
\begin{prop} There exists the rational expectations equilibrium $\Delta A_0=A_0-A_0=0$ in which agents' price forecasts are
correct $\hat{P}=P_0$. There exists a second equilibrium $\Delta
A_1=A_1-A_0=-\frac{(1-\phi_A)}{\tau}<0$ in which
$\hat{P_1}\gtreqqless P_1$.\label{p2}
\end{prop}
Both equilibria in Proposition \ref{p2} are self-fulfilling. In
the perfect foresight equilibrium, no firm deviates from the
equilibrium supply $A=A_0$, and thus there is no need for agents
to rely on polynomial approximations: Producers \emph{know} the
price $\phi(A_0)$. The opposite is the case in the second
equilibrium: Once firms supply $A_1\neq A_0$, they are uncertain
as to the equilibrium price, $\phi(A_1)=\phi(A_0+\Delta A_1)$,
which they can only approximate as $\phi(A_0)+\phi_A(A_0)\Delta
A$. Moreover, the error of this approximation, $(A-A_0)^2$, grows
the more agents deviate from supplying $A_0$. That is, once firms
deviate from the rational expectations equilibrium, they find it
harder to estimate future prices and thus they are incentivised to
deviate even further until a new equilibrium is reached. In this
equilibrium, firms cannot forecast prices accurately, and thus
they choose to produce a small number of goods, at a low marginal
cost, which provides a margin of safety.

As we argued earlier, this interdependence between aggregate
output and the individual firm's ability to understand the
environment that it operates in is a crucial aspect in most
crises: Once consumers and investors change their behavior, they
find themselves in an environment that is hard to understand, and
they hold back on investment and consumption decisions waiting for
the ``dust to settle''. In the current interpretation, by cutting
production, agents put \emph{themselves} into ``uncharted
territory''. This aspect is, by assumption, not captured in
environments where agents can compute the entire demand function,
respectively, solve the model as in Lemma \ref{p1}. Before we
discuss the scope for government to correct such ``glitches" in
output, which turns out to be limited, we make one remark: The
model's coefficients $\phi_A(A_0), \tau$ can be chosen such that
$\Delta A_1=A_1-A_0=-\frac{(1-\phi_A)}{\tau}<0$ is arbitrarily
small. That is, both equilibria in Proposition \ref{p2} exist even
if the the first-order Taylor-series approximation (\ref{hat}) is
of very good quality, i.e., if the
error term is of order $\mathcal{O}(\Delta A^2)$. 

\subsection{Parameter Changes}\label{s3}

We augment demand $P=\phi(A;b)$ to incorporate an exogenous
parameter $b$. This parameter is assumed to increase demand
$\phi_b=\phi_b(A;b)>0$ for every $A$. This parameter may be seen
as government demand or money supply.\footnote{Alternatively, as
we discuss in Section \ref{ss16}, the model may be interpreted as
the capital market of an overlapping generations economy, where
$a_i$, $A$, $\phi()$ are, respectively, individual savings and
aggregate savings, and $\phi(A)$ is the marginal product of
capital that agents \emph{expect} to receive on their savings.
Finally, $b$ may be seen as public debt and $A_0$ as steady state
capital.} In this interpretation, the following section identifies
the multiplier effect that obtains once agents need to rely on
approximations to anticipate the consequences of policy
interventions.

We begin with a benchmark model where agents are computationally
unconstrained. Second, we study the model with friction. Comparing
both settings shows that parameter increases, which increase
demand and equilibrium output in a model with unconstrained
agents, can reduce economic activity if firms are computationally
constrained. Put differently, parameter changes, in particular if
they are large, can put agents into ``uncharted territory", and
incentivise them to cut, rather than increase, output.

\subsubsection{Comparative statics without friction}\label{nof}
Recalling our augmented demand function:
\begin{eqnarray} P=\phi(A;b_0), \quad \phi_A<0, \quad
\phi_{AA}\gtrless0, \quad \phi_b>0, \quad
\phi(0;b)>0,\label{be1}\end{eqnarray} firms can anticipate the
equilibrium price $P$ correctly, if they are computationally
unconstrained as in Lemma \ref{p1}. Hence, they choose a
production schedule $a_i^*$ which maximizes profits
$a^*_i=\underset{a_i}{\arg\max}\Big\{\pi_i=a_iP-\frac{1}{2}a_i^2\Big\}$.
Aggregate supply is thus
\begin{eqnarray}
A=\int_{[0,1]}a^*_idi=P.\label{be33}\end{eqnarray} Taken together
(\ref{be1}) and (\ref{be33}) yield a unique equilibrium $P_0,A_0$
for \emph{every} given exogenous parameter $b_0$. Once the
parameter changes from $b_0$ to $b_1=b_0+\Delta b$, the price is
again correctly anticipated as the unique solution $P_1,A_1$ to
the equations $P=A$ and $P=\phi(A;b_0+\Delta b)$.

How would an actual human being, or an economic researcher, try to
think about the impact of the parameter change? The outside
researcher, who uses textbook methods to study how changes in the
exogenous parameter from $b_0$ to $b_1$ change output and price,
cannot compute $A$ and $P$ explicitly. Instead, he will
approximate the model's comparative statics. That is, he will
differentiate (\ref{be1}) and (\ref{be33}):
\begin{eqnarray} \Delta P\approx\phi_A(A_0;b_0)\Delta A+\phi_b(A_0;b_0)\Delta b ,\quad \Delta A=A-A_0,\quad \Delta b=b_1-b_0,\label{de41.1}\end{eqnarray}
\begin{eqnarray} \Delta A=\Delta P.\label{be4}\end{eqnarray}
Combining (\ref{de41.1}) and (\ref{be4}) yields the model's
comparative statics:
\begin{lemma}\label{L2} Exogenous parameter variations $\Delta b$ change output (and price) according to
$\frac{\Delta A}{\Delta b}\approx\frac{\phi_b}{1-\phi_A}>0$ and
$\frac{\partial A}{\partial b}=\underset{\Delta b\rightarrow
0}{\lim}\frac{\Delta A}{\Delta
b}=\frac{\phi_b}{1-\phi_A}>0$.\end{lemma}

That is, an outside observer/analyst would use a polynomial
expansion of $A=\phi(A;b), P=\phi(A;b)$ to \emph{approximate} the
impact of an exogenous parameter change as in Lemma \ref{L2}. In
the following section, we assume that firms themselves make such
``polynomial inference" using such an approximation to anticipate
the consequences of parameter changes.

\subsubsection{Comparative statics with computationally constrained
agents}\label{con} Using a first-order expansion, around
$A_0=\phi(A_0;b_0)$, agents approximate the equilibrium price:
\begin{eqnarray} \hat{P}=\phi_0(A_0,b_0)+\phi_A\Delta A+\phi_b\Delta b,\quad \Delta A=A-A_0,\quad \Delta b=b_1-b_0. \label{ne22}\end{eqnarray}
That is, agents have to incorporate two aspects in their demand
forecast: (i) the direct effect of the parameter change $\Delta b$
and (ii) the equilibrium response of all agents who deviate
$\Delta A\neq 0$ from their usual supply choice. As before, agents
discount the price estimate since they do not know how curvature
terms of demand $\phi_{AA}$, $\phi_{bb}$, and $\phi_{Ab}$ affect
prices:\footnote{In Appendix \ref{A1} we discuss an alternative
error term $\max[\Delta A^2,\Delta b^2,\Delta A\Delta b]$, where
agents are either concerned about miscalculating the parameter
change's impact, $\Delta b>\Delta A$, or the other agents'
reaction $\Delta A>\Delta b$ to the parameter change.}
\begin{eqnarray}&&\pi_i=a_i\hat{P}-a_i(\tau_1 \Delta A^2+\tau_2\Delta b^2+\tau_3|\Delta
A||\Delta b|)-\frac{1}{2}a_i^2\nonumber\\
&&a_i^*=\hat{P}-\tau_1 \Delta A^2-\tau_2\Delta b^2-\tau_3|\Delta
A||\Delta b|\label{ne24}\quad \tau_i\geq 0, \quad i=1,2,3.
\end{eqnarray}

To find the equilibria associated with (\ref{ne22}) and
(\ref{ne24}), it is useful to distinguish cases where $\Delta
A\geq0$ from cases where $\Delta A\leq0$.

We begin by looking for equilibria where $\Delta A\geq0,\Delta
b\geq0$. If $\Delta A\geq0$ and $\Delta b\geq0$ then supply equals
approximate demand (\ref{ne24}), when:
\begin{eqnarray} A=\phi_0(A_0,b_0)+\phi_A\Delta A+\phi_b\Delta b-\tau_1 \Delta A^2-\tau_2\Delta b^2-\tau_3\Delta
A\Delta b, \nonumber\end{eqnarray} and thus:
\begin{eqnarray}\label{ne26} \Delta A_{1,2}=-\frac{1-\phi_A+\tau_3\Delta b}{2\tau_1}
\pm\sqrt{(\phi_b-\tau_2\Delta b)\frac{1}{\tau_1}\Delta
b+\Big(\frac{1-\phi_A+\tau_3\Delta b}{2\tau_1}\Big)^2}.
\end{eqnarray} Combining the two equilibrium candidates in (\ref{ne26}) with our initial assumption $\Delta
A\geq0$, we have:
\begin{lemma}\label{L3} If and only if $\Delta
b<\frac{\phi_b}{\tau_2}$, there exists an equilibrium in which,
compared to the perfect foresight equilibrium, production (and
price) are increased:\\ $\Delta A_{1}=-\frac{1-\phi_A+\tau_3\Delta
b}{2\tau_1}+\sqrt{(\phi_b-\tau_2\Delta b)\frac{1}{\tau_1}\Delta
b+\Big(\frac{1-\phi_A+\tau_3\Delta b}{2\tau_1}\Big)^2}>0$. At the
margin, increases in the parameter increase income if
$\frac{\partial\Delta A}{\partial \Delta
b}=-\frac{\tau_3}{2\tau_1}+1/2\frac{\phi_b\frac{1}{\tau_1}-2\frac{\tau_2}{\tau_1}\Delta
b+2\frac{\tau_3}{\tau_1}\Big(\frac{1-\phi_A+\tau_3\Delta
b}{2\tau_1}\Big)}{\sqrt{(\phi_b-\tau_2\Delta
b)\frac{1}{\tau_1}\Delta b+\Big(\frac{1-\phi_A+\tau_3\Delta
b}{2\tau_1}\Big)^2}}>0$.
\end{lemma}
\begin{proof} Follows directly from (\ref{ne26}).
\end{proof}

To interpret the equilibrium in Lemma \ref{L3}, we study how it
corresponds to the REE of Proposition \ref{p2}. That is, we note
that $\underset{\Delta b\rightarrow 0}{\lim}\Delta A_{1}(\Delta
b)=0$, i.e., the equilibrium quantity $A_1$ converges to the REE
quantity  $A_0$ as $b\rightarrow b_0$. On the contrary, for large
changes $\Delta b>0$, the equilibrium loses its RE character as
agents do not precisely know how demand is impacted by the
exogenous change. Such large parameter changes have an ambiguous
effect on output, which is captured by the term
$(\phi_b-\tau_2\Delta b)\Delta b$. On the one hand, agents
extrapolate the increase in demand $\phi_b\Delta b>0$. At the same
time, agents cannot rule out that too large an increase might
eventually prove counterproductive $-\tau_2\Delta b^2$. That is,
large, unprecedented changes in the model's structure render
agents' polynomial approximations unreliable, and put them into
``uncharted territory".

Lemma \ref{L3} thus features four policy regimes. First, if the
parameter change is (infinitesimally) small, the agents'
polynomial approximations are of high quality. In this regime,
policy is as effective as in the model of Section \ref{nof}, Lemma
\ref{L2}, where agents are computationally unconstrained. That is,
the model's multiplier is given by $\frac{\partial\Delta
A}{\partial \Delta b}_{|\Delta b=0}=\frac{\partial A}{\partial
b}=\frac{\phi_b}{1-\phi_A}>0$. Second, there is an intermediate
region where $\Delta b\in[0,\Delta b_1]$, and policy changes have
a positive effect at the margin, $\frac{\partial \Delta
A}{\partial \Delta b}>0$. These marginal returns, however, are
diminishing. Third, there is a region $\Delta b\in[\Delta
b_1,\Delta b_2]$ where agents find themselves in ``uncharted
territory" and start to cut production $\frac{\partial \Delta
A}{\partial \Delta b}<0$. Finally, in the extreme case where
$\Delta b\geq \frac{\phi_b}{\tau_2}$, an equilibrium $\Delta A>0$
cannot exist.


Regarding the remaining equilibria, where $\Delta A\leq0$, we
recall (\ref{ne22}) and (\ref{ne24}), and note that $-|\Delta
A|=\Delta A$. Accordingly, there are two candidates
\begin{eqnarray}\label{ne27} \Delta A_{2,3}=-\frac{1-\phi_A-
\tau_3\Delta b}{2\tau_1}\pm\sqrt{(\phi_b-\tau_2\Delta
b)\frac{1}{\tau_1}\Delta b+\Big(\frac{1-\phi_A-\tau_3\Delta
b}{2\tau_1}\Big)^2}.
\end{eqnarray}
In view of (\ref{ne27}), if $\Delta b<\frac{\phi_b}{\tau_2}$,
there exists exactly one equilibrium, in which $\Delta
A_{2}=-\frac{1-\phi_A- \tau_3\Delta
b}{2\tau_1}-\sqrt{(\phi_b-\tau_2\Delta b)\frac{1}{\tau_1}\Delta
b+\Big(\frac{1-\phi_A-\tau_3\Delta b}{2\tau_1}\Big)^2}<0$ such
that economic activity is lower than in the REE. For large
parameter changes, $\Delta b>\frac{\phi_b}{\tau_2}$, there can
exist up to two equilibria in which economic activity is low.
Combining these observations with Lemma \ref{L3} yields:

\begin{coro} Large parameter changes $\Delta b>\frac{\phi_b}{\tau_2}$
preclude the existence of equilibria with output $A\geq
A_0$.\end{coro}


\section{Extensions}

So far agents were assumed to know the rational expectations
equilibrium $A_0,\phi(A_0)$. In Section \ref{s5}, we study a
dynamic setting where agents learn different points on the demand
curve over time. In turn, we examine how the economy converges to
the REE. Second, in our baseline setting, all agents know the same
point on the demand curve. Price forecasts and supply decisions
are therefore the same across agents. Once different agents know
different pieces of the demand curve, this is no longer true.
Rather than knowing each other's price forecasts and supplies,
agents have to estimate price \emph{and} supply simultaneously. In
Section \ref{s4}, we extend our model to incorporate such
asymmetric information in a manner which is akin to Bayesian
inference.

\subsection{Learning}\label{s5}

We abstracted from the fact that agents may learn from past
mistakes, i.e., suboptimal production choices that were based on
incorrect price estimates. One would imagine that they memorize
these mistakes, or the observation that a quantity $A_1$ is
associated with an observable price $P_1=\phi(A_1)$. Under our
current assumptions on the demand function, this price differs
from the estimated price $\hat{P_1}$. Hence, agents would not
supply $A_1$ again. Second, if agents are computationally
constrained to the use of polynomials, how do they find the
perfect foresight equilibrium in the first place? This section's
main observation is that agents will learn the REE over
time.


Agents who sell repeatedly into the market will, over time
$t=0,1,2,3...$, observe an increasing number of points on the
demand curve. Regarding these points, $P_t=\phi(A_t)$, we assume
that agents also learn demand's slope $\phi_A(A_t)$ once a
quantity $A_t$ is marketed. Given past observations $A_t,
t=0,1,2,3,...,T$ agents can refine their price estimate as:
\begin{eqnarray} \hat{P}_{T+1}=\phi(A^*)+\phi_A(A^*)(A_{T+1}-A^*), \quad A^*=\underset{A_t}{\arg\min}\Big\{|A_{T+1}-A_t|\Big\}
 \quad t=0,1,2,3...T \label{l1}\end{eqnarray}
That is, to estimate prices, they select from the set of known
points $\{A_t\}_{t=0}^{T}$ the point $A^*$, which is closest to
the future supply $A_{T+1}$. Put differently, they use the
observation $A^*$ from the past, which is most similar/closest to
the situation they are trying to make inference on. In turn,
agents $i$ choose supply
\begin{eqnarray} a_i=\hat{P}_{T+1}-(A_{T+1}-A^*)^2.\label{l11} \end{eqnarray}
Hence, for a given $A^*$, there are two equilibrium candidates
\begin{eqnarray} A_{T+1}=A^*-\frac{1-\phi_A(A^*)}{2}\pm\sqrt{(\phi(A^*)-A^*)+\Big(\frac{1-\phi_A(A^*)}{2}\Big)^2}. \label{l2}\end{eqnarray}

To show that (\ref{l1}) and (\ref{l2}) ensure that agents learn
the REE equilibrium $A_0,\phi(A_0)$, we proceed in two steps.
First, we study the case where demand $\phi()$ is a convex
function. In this case, convergence to the REE can be studied in
terms of a simple first-order difference equation. Second, for the
remaining cases, we give an indirect argument in Appendix
\ref{ALR}.

\subsubsection{Convex Demand}\label{cd} Without loss of generality, we
assume that agents start with a prior $\mu,\phi(\mu)$, $\mu<A_0$.
Moreover, we focus on the ``+" roots of (\ref{l2}). For convex
demand, we now show that (\ref{l1}) and (\ref{l2}) imply a
first-order difference equation for supply:
\begin{eqnarray} A_{T+1}=A_T-\frac{1-\phi_A(A_T)}{2}+\sqrt{(\phi(A_T)-A_T)+\Big(\frac{1-\phi_A(A_T)}{2}\Big)^2}. \label{l244}\end{eqnarray}
First, we note that (\ref{l244}) indicates that agents, in
Marshallian fashion, increase supply, such that $A_{T+1}>A_T$, if
the marginal revenue $\phi(A_T)$ exceeds the marginal cost of
production $A_T$. Moreover, from Lemma \ref{p1}, we know that
there exists only one level of supply, namely $A_0$, where
$A=\phi(A)$. It follows that (\ref{l244}) has a unique steady
state at the point where $A_{T+1}=A_T=A_0$. This steady state
equilibrium is stable due to our assumption that $\phi$ is
downward-sloping: For $A_T<A_0$, we have $\phi(A_T)-A_T>0$ and
thus $A_{T+1}>A_T$. At $A_0$, the system is locally stable since
$\frac{\partial A_T}{\partial A_{T+1}}=_{|A_T=A_0}0\in(-1,1)$. To
complete the argument, we note that the sequence
$\{A_t\}_{t=1}^{T}$ is strictly increasing, and, due to our
convexity assumption on $\phi$, that $A_t\leq A_0\forall
t=0,1,2....T$.\footnote{To see this, recall (\ref{l1}) and
(\ref{l11}), which imply
$A_{T+1}=\phi(A_T)+\phi_A(A_T)(A_{T+1}-A_T)-(A_{T+1}-A_T)^2\leq
\phi(A_T)+\phi_A(A_T)(A_{T+1}-A_T)$. At the same time, convexity
of $\phi$ implies: $\phi(A_{T+1})=\phi(A_T+A_{T+1}-A_T)\geq
\phi(A_T)+\phi_A(A_T)(A_{T+1}-A_T)$. Taking both inequalities
together, we have $A_{T+1}-\phi(A_{T+1})\leq 0$, respectively,
$A_{T+1}\leq A_0$. Where $A_{T+1}\leq A_0$, follows from $\phi$
being downward-sloping and $A=\phi(A)$ at $A=A_0$. Hence, if we
start at a point $\mu-\phi(\mu)<0$, this implies that $A_{T+1}\leq
A_0\forall T$. Finally, as mentioned before, if demand is
non-convex, $A_t$ can overshoot $A_0$. In that case we require an
additional argument, which we give in Appendix \ref{ALR}.} That
is, $A_T$ always adjusts towards, but not beyond, $A_0$. Hence,
according to (\ref{l1}), agents will always use the information
that they learned in the previous period, when a quantity $A_{T}$
was marketed, to think about $A_{T+1}$. This last property allows
us to study convergence in terms of the first-order difference
equation (\ref{l244}).


\subsection{Asymmetric Information}\label{s4}

One may suspect that heterogeneity in information might mitigate
the possibility of multiple equilibria that we
emphasize.\footnote{See, e.g., \citet{Mor98} for a coordination
problem where asymmetric information selects unique equilibria in
an economy with a continuum of players.} Moreover, one might
expect that dispersion of private information induces some agents
to supply too little and others too much such that, on average,
errors cancel and supply might actually be at an efficient
level.\footnote{See \citet{Gall07}, and \citet{Gro76} for such
wisdom of the crowd effects.} Regarding these conjectures, we find
that (i) multiplicity carries over to the case with dispersed
private information and (ii) that dispersed information tends to
amplify (dampen) supply if demand is concave (convex).

Each agent $i\in[0,1]$ is assumed to know the selling price
$\phi(A_i)$ and demand's first derivative $\phi_A(A_i)$ of a
particular supply $A_i\in[0,\infty]$. Agents are distributed over
these points according to an integrable density function $f()$.
For simplicity, we normalize agents' discount rate to $\tau=1$.

Conditional on information $A_i,\phi(A_i),\phi_A(A_i)$ agent $i$
supplies:
\begin{eqnarray} a^*_i|A_i=\hat{P}|A_i-(\hat{A}|A_i-A_i)^2, \label{ha1} \end{eqnarray}
where $\hat{P}|A_i$ and $\hat{A}|A_i$ are agent $i's$ price and
supply forecasts conditional on knowing demand $\phi(A_i)$ at
point $A_i$. The polynomial estimates for price and quantity are:
\begin{eqnarray} \hat{P}|A_i=\phi(A_i)+\phi_A(A_i)(\hat{A}|A_i-A_i), \label{ha2} \end{eqnarray}
and
\begin{eqnarray} \hat{A}|A_i=\int_{[0,1]}(a_j|A_j)|A_idj. \label{ha3} \end{eqnarray}
Where (\ref{ha3}) reflects that agent $i$ uses his information at
point $A_i$ to infer the information and thus the supply of the
other agents who know a different point $A_j$. That is,
\emph{agent $i$ knows that agent $j$ observes a point on the same
demand curve and thus he uses a polynomial expansion around
$\phi(A_i)$ to estimate the information $\phi(A_j),\phi_A(A_j)$
that player $j$ receives}. Based on this reasoning, $i$ can
construct an estimate for the other players' price estimates,
which he needs to calculate aggregate supply. Agent $i's$ price
and supply estimates are thus given by the simultaneous solution
of (\ref{ha2})-(\ref{ha3}). In turn, he can choose supply
(\ref{ha1}). We solve the model in Appendix \ref{A2.1} using a
guess-and-verify approach.

These solutions yield two main insights. First, as in Proposition
\ref{p2}, multiple equilibria exist due to the interaction between
agents' ability to forecast the equilibrium and aggregate supply.
Second, unlike the earlier model, where information was symmetric,
we show in Lemma \ref{Lem4} that output is depressed across
\emph{all} equilibria since agents systematically underestimate
demand if the true demand function is convex. Moreover, the
marginal cost of production differs among producers, and thus
output, in addition to being low, is produced inefficiently.


\section{Interpretations}\label{ss16} In this section, we
reinterpret our model in terms of the three macroeconomic
workhorse frameworks: (i) aggregate search models of the
\citet{Dia82} type, (ii) Life-cycle savings models of the
\citet{Dia65} type, and (iii) models of supply and demand as in
\citet{Sam09} and \citet{Mas95}.

\paragraph{Search:} In the context of the \citet{Dia82}, p. 887, model, agents face
the following choice problem:
\begin{eqnarray} \max_{a_i}\{a_i\phi(A;b_0)-f(a_i)\},\quad a_i\in\mathcal{R_+}.  \nonumber\end{eqnarray}
Where $a_i$ is individual $i's$ output choice in Period 1, and
$\phi(A;b_0)$ is the probability with which agents find a trading
partner in Period 2. If a trading partner is found, agent $i$ can
sell/exchange goods at price one. Finally, $f(a_i)$ is the cost of
production.

The chance of finding a trading partner, $\phi(A;b_0)$, is an
increasing function in aggregate economic activity
$A=\int_{[0,1]}a_idi$ and government demand $b$. \citet{Dia82}
shows that such an economy can have multiple REE equilibria, which
we call, say, $A_0, A_1$. Suppose now that agents know one of
these REE, e.g., $A_0$; then, if they are computationally
constrained, as in the present paper, they would need to use a
polynomial expansion to compute the probability $\phi(A_0+\Delta
A)$ of finding a trading partner that would prevail once agents
collectively deviate $\Delta A$ from producing $A_0$. The same
applies to the evaluation of the exogenous policy parameter $b$,
which may, unlike in \citet{Dia82}, result in a negative
multiplier effect, as discussed in Section \ref{s3}.

\paragraph{Savings and Investment:} In the context of the \citet{Dia65} model, $\phi()$ may be interpreted
as a component to agents' consumption savings problem:
\begin{eqnarray} \max_{s_{t},c^1_t,c^2_{t+1}}U(c_t^1,c_{t+1}^2)
\quad s.t. \quad c^1_t=w_t-s_t, \quad
c^2_{t+1}=s_t(1+r_{t+1}),\quad c^1_t>0,
c^2_{t+1}>0,\label{3}\end{eqnarray} where factor prices are
functions of the prevailing capital-labor ratios $k_t$ and
$k_{t+1}$:
\begin{eqnarray} r_{t+1}=f'(k_{t+1}),\quad w_t=f(k_t)-f'(k_t)k_t.\nonumber\end{eqnarray}
To make choice (\ref{3}), agents have to form expectations
regarding equilibrium interest $r_{t+1}=f'(k_{t+1})$. In
equilibrium, the life-cycle savings condition, $(1+n)k_{t+1}=s_t$,
relates savings and capital; $n$ representing the exogenous rate
of population growth.

Suppose now that the economy is initially in a steady state at
$k_0,s_0$. To compute future interest rates, agents have to
compute $r_{t+1}=f'(k_{t+1})=f'(\frac{s_t}{1+n})=:\phi(s_t;n)$.
Once again, if agents know the prevailing interest in the steady
state, they have to engage in polynomial expansions to form price
expectations $\hat{r}_{t+1}$ to compute the interest rate that
obtains in the (temporary) equilibria that obtain once agents
choose savings $s_{t}=s_0+\Delta s_{t}$.

This argument extends to the case where agents supply labor in
both periods. In that case, to make their savings decision, agents
have to (i) approximate the interest rate and (ii) the
second-period wage rate. That is, they have to approximate the
\citet{Sam62} neoclassical factor-price frontier,
$w_{t+1}=\xi(r_{t+1}(k_{t+1}(s_t)))$, which relates wages to
interest, interest to the capital intensity, and finally the
capital intensity to savings.

\paragraph{Supply and Demand:} Our lead interpretation was that of a simplified $A^D,A^S$
setting. Taking this perspective, we suggest one reason why demand
analysis may be computationally complicated for firms. Suppose
demand is given by $A^D=\xi(A;L(A))$, where $\xi$ represents
demand, which is downward-sloping in supply $\frac{\partial
\xi}{\partial A}<0$. At the same time, cuts in production $A$
reduce employment $L(A)$ and demand, i.e., $\frac{\partial
\xi}{\partial L}\frac{\partial L}{\partial A}>0$. Accordingly,
agents grapple with the question whether the function
$\phi(A):=\xi(A;L(A))$ is upward or downward sloping once supply
falls into regions which they do not know from past experience.

Similarly, suppose that demand concerns a vector
$\textbf{A}\in\mathcal{R}^L$ of goods, involving substitutes and
complements, as in \citet{Mas95}, pp. 599-641. Suppose that a firm
produces a particular good $a_l$, then, if all other firms in the
economy, supplying the various other goods, change their behavior,
it has to analyze how changes in the supplies and prices for the
other goods $L\setminus l$, influence demand, and thus price, for
good $l$. In turn, if this firm is exceptionally well-informed, it
might know the entries of the economy's Jacobian matrix. However,
it need not know demand's second- and cross-derivatives, which
once again makes it difficult to compute demand correctly.




\section{Conclusion}\label{s7}

Economists' forecasting record suggests that it is difficult to
compute future economic events. The current model recognizes this
and assumes that agents cannot compute exact numeric values for
future equilibrium outcomes such as prices.

The model's key feature is that the precision with which agents
can approximate future equilibrium prices depends on the level of
aggregate economic activity, and is thus endogenous. This
interdependence between aggregate output and an individual firm's
ability to forecast the price at which it can sell its output
gives rise to equilibria in which economic activity is
inefficiently low. Such equilibria may be interpreted as
``glitches" of the overall economy. During such a glitch, agents
collectively reduce economic activity. This change in behavior
makes it difficult to forecast the resulting equilibrium, which,
in turn, justifies the initial output cut. For similar reasons we
also find that the scope for government to correct such
``glitches" in output is limited: Interventions, which would
unambiguously increase output in a frictionless economy, can make
it harder for firms to predict future equilibria and reduce output
even further. Our model therefore captures the common place
observation that large parameter changes, such as unprecedented
quantitative easing, can put agents into ``uncharted territory".


The particular form in which equilibria obtain depends on the
assumption that agents use the same Taylor-series expansions that
an economic researcher, who applies standard textbook methods,
would use. More sophistication on the part of agents will
undoubtedly change the specific form and number of equilibria.
However, it appears unlikely that the precision with which future
equilibrium outcomes are approximated can ever be entirely
independent of the overall level of economic activity, which is
what our findings rely on.

\newpage
\begin{appendix}

\section{Second-order Expansions}\label{A2} In this appendix, we derive our results for a setting where agents know of demand's
first and second derivatives. They can thus use second-order
expansions to estimate prices:
\begin{eqnarray} \hat{P}=\phi(A_0)+\phi_A\Delta A+\frac{1}{2}\phi_{AA}\Delta A^2, \label{sec1} \end{eqnarray}
We study the equilibria that emerge once agents are averse
$\tau>0$ to the third-order error. Setting $\tau=0$, we obtain the
equilibria that emerge if agents are indifferent regarding errors.
Using the estimate (\ref{sec1}), firms choose a profit-maximizing
quantity:
\begin{eqnarray} a^*_i=\underset{a_i}{\arg\max}\Big\{\pi=a_i\hat{P}-\frac{1}{2}a_i^2-a_i\tau|\Delta A^3|\Big\},\quad \tau \geq 0.\nonumber  \end{eqnarray}
Individual and aggregate supply are thus
\begin{eqnarray} a_i^*=\hat{P}-\tau|\Delta A|^3,\quad A=\int_{[0,1]}a_i^*di=\hat{P}-\tau|\Delta A|^3.  \label{p31} \end{eqnarray}
Combining (\ref{sec1}) and (\ref{p31}), we obtain:
\begin{eqnarray} A+\tau|\Delta A|^3=\phi(A_0)+\phi_A\Delta A+\frac{1}{2}\phi_{AA}\Delta A^2. \label{1000} \end{eqnarray}
Recalling that $A_0=\phi(A_0)$, (\ref{1000}) writes:
\begin{eqnarray} \tau|\Delta A|^3-\frac{1}{2}\phi_{AA}\Delta A^2+(1-\phi_A)\Delta A=0. \label{1001} \end{eqnarray}
If $\tau=0$ we have:
\begin{prop}\label{p3} If $\tau=0$, there exists the perfect foresight equilibrium $\Delta A_0=0$ and a second equilibrium $\Delta
A_1=\frac{1-\phi_A}{\frac{1}{2}\phi_{AA}}$.\end{prop}

According to Proposition \ref{p3}, economic activity in the
polynomial equilibrium exceeds activity in the perfect foresight
equilibrium if the demand function's second derivative indicates
that demand may rebound $\frac{1}{2}\phi_{AA}>0$ once supply
exceeds $A_0$. A negative second derivative $\frac{1}{2}\phi_{AA}$
depresses output in the same manner as the discount factor $\tau$
in Proposition \ref{p2}. As before, estimated demand can meet
supply more than twice if agents discount their price estimate:
\begin{prop}\label{p4} If $\tau>0$,
there exist at least two equilibria: $\Delta A_0=0$, and $\Delta
A_1=-\frac{1}{4\tau}\phi_{AA}-
\sqrt{(\frac{1}{4\tau}\phi_{AA})^2+\frac{1-\phi_A}{\tau}}<0$, in
which economic activity is depressed. If $\phi_{AA}>0$, and
$(\frac{1}{4\tau}\phi_{AA})^2>\frac{1-\phi_A}{\tau}$, there may
exist two additional equilibria where economic activity is
elevated $\Delta A_{2,3}=\frac{1}{4\tau}\phi_{AA}\pm
\sqrt{(\frac{1}{4\tau}\phi_{AA})^2-\frac{1-\phi_A}{\tau}}>0$.\end{prop}
\begin{proof} Using (\ref{1001}), we find the perfect foresight equilibrium
$\Delta A_0=0$. To identify the remaining equilibria, we
distinguish cases (i) $\Delta A<0$ and (ii) $\Delta A>0$.

1.) Assuming $\Delta A<0$: we note that $|\Delta A^3|=-\Delta
A^3$. Dividing by $\Delta A$, we find that (\ref{1001}) has roots:
$\Delta A_{1,2}=-\frac{1}{4\tau}\phi_{AA}\pm
\sqrt{(\frac{1}{4\tau}\phi_{AA})^2+\frac{1-\phi_A}{\tau}}$.
However, only one root satisfies the initial assumption $\Delta
A<0$. That is, $\Delta A=-\frac{1}{4\tau}\phi_{AA}-
\sqrt{(\frac{1}{4\tau}\phi_{AA})^2+\frac{1-\phi_A}{\tau}}<0$,
regardless of the sign of $\phi_{AA}$. The other root $\Delta
A=-\frac{1}{4\tau}\phi_{AA}+
\sqrt{(\frac{1}{4\tau}\phi_{AA})^2+\frac{1-\phi_A}{\tau}}>0$ is
positive since $(\frac{1}{4\tau}\phi_{AA})^2>0$, $\phi_A<0$,
$\frac{1-\phi_A}{\tau}>0$. It thus violates the initial assumption
$\Delta A<0$.

2.) Assuming $\Delta A>0$: we note that $|\Delta A^3|=\Delta A^3$
and find that (\ref{1001}) has two real roots $\Delta
A_{1,2}=\frac{1}{4\tau}\phi_{AA}\pm
\sqrt{(\frac{1}{4\tau}\phi_{AA})^2-\frac{1-\phi_A}{\tau}}$ if
$(\frac{1}{4\tau}\phi_{AA})^2>\frac{1-\phi_A}{\tau}$. Both of
these roots are negative if $\phi_{AA}<0$ violating the assumption
$\Delta A>0$. Hence, $\phi_{AA}>0$ is a sufficient condition for
$\Delta A>0$ equilibria to exist.
\end{proof}

The polynomial equilibria in propositions \ref{p2}-\ref{p4} have
in common that they originate from a coordination problem: In
their price forecasts, each agent takes the overall supply $A$ as
given. In equilibrium, however, aggregate supply depends on
agents' price forecasts. Hence, the price forecast itself is an
equilibrium outcome. At this point, it is clear that higher-order
polynomials yield even more equilibria, and that propositions
\ref{p2}-\ref{p4} carry over qualitatively once we introduce
demand $\Phi(\textbf{A})$, for a vector $\textbf{A}$ of
goods.

\section{Alternative Discounting}\label{A1} The price
estimate is as before:
\begin{eqnarray} \hat{P}=\phi_0(A_0,b_0)+\phi_A\Delta A+\phi_b\Delta b,\quad \Delta A=A-A_0,\quad \Delta b=b-b_0. \label{e22}\end{eqnarray}
The model differs in agents' discounting:
\begin{eqnarray}&&\pi_i=a_i\hat{P}-a_i\max[\Delta A^2,\Delta b^2,\Delta A\Delta
b]-\frac{1}{2}a_i^2\nonumber\\
&&a_i^*=\hat{P}-\max[\Delta A^2,\Delta b^2]\label{e24}
\end{eqnarray}
According to (\ref{e24}) there are two regimes. First, agents fear
that they miscalculate the impact of the exogenous parameter
change in case $\Delta b>\Delta A$. Second, agents fear that they
misjudge the other agents' reaction to the exogenous parameter
variation $\Delta A>\Delta b$. To analyze the equilibrium outcomes
associated with (\ref{e22}) and (\ref{e24}), we distinguish cases
where the parameter change is relatively large, $|\Delta
b|>|\Delta A|$, from cases where its impact is relatively small,
$|\Delta b|<|\Delta A|$. Note that $|\Delta A|$ is a function of
$|\Delta b|$. That is, we start with the assumption that, e.g.,
$|\Delta b|>|\Delta A|$ and solve for the equilibrium $\Delta A$.
In turn, we check whether the initial assumption $|\Delta
b|>|\Delta A|$ is correct.
\begin{prop} If $|\frac{\phi_b-\Delta b}{1-\phi_A}|<1$, there exists only one equilibrium $\Delta A=\Big(\frac{\phi_b}{1-\phi_A}-\frac{\Delta
b}{1-\phi_A}\Big)\Delta b$ in which $|\Delta b|>|\Delta A|$. In
this equilibrium $\underset{b\rightarrow b_0}{\lim}\frac{\Delta
A}{\Delta b}=\frac{\phi_b}{1-\phi_A}=\frac{\partial A}{\partial
b}_{|A_0,b_0}$.\label{p4}
\end{prop}
\begin{proof} Individual supply is $a_i=\hat{P}-\max[\Delta A^2,\Delta b^2]$ under the assumption that $|\Delta b|>|\Delta A|$, we have
$a_i=\hat{P}-\Delta b^2$. Aggregate supply is thus
$A=\int_{i\in[0,1]}a_i=\hat{P}-\Delta b^2$. Equilibrium requires
$A=A_0+\phi_A\Delta A+\phi_b\Delta b-\Delta b^2$, respectively,
$\Delta A (1-\phi_A)+\Delta b(\Delta b-\phi_b)=0$. Solving yields
$\Delta A=\frac{\phi_b-\Delta b}{1-\phi_A}\Delta b$. It remains to
note that $|\frac{\phi_b-\Delta b}{1-\phi_A}|<1$ ensures that
$|\Delta b|>|\Delta A|$.\end{proof}

As before, Proposition \ref{p4}, $\Delta
A=\Big(\frac{\phi_b}{1-\phi_A}-\frac{\Delta
b}{1-\phi_A}\Big)\Delta b$, shows that agents extrapolate the
positive impact that a parameter change $\Delta
A=\frac{\phi_b}{1-\phi_A}$. At the same time, they are facing
increased uncertainty as to the actual price at which their
products sell $\Delta A=-\frac{\Delta b}{1-\phi_A}\Delta b$. A
large parameter change where $\Delta b>\phi_b$ thus reduces
economic activity as the uncertainty that it creates outweighs the
expansive effect $\phi_b>0$.

This leaves us with equilibria where agents are more concerned
about the potential error associated with aggregate supply
changes. In these cases, agents are primarily afraid that they
forecast prices incorrectly due to the change $\Delta A$. For
cases where $|\Delta b|<|\Delta A|$ we have:

\begin{prop} There exists an upper bound $\Delta b_1>0$ and an equilibrium where
$\Delta A_{1}=-\frac{1}{2}(1-\phi_A)-\sqrt{\phi_b\Delta
b+(\frac{1}{2}(1-\phi_A))^2}<0$ and $|\Delta b|<|\Delta A|$, if
$\Delta b\in[0,\Delta b_1]$. If $\frac{\phi_b}{1-\phi_A}>1$ there
exists an upper bound $\Delta b_2>0$ and a second equilibrium
$\Delta A_2=-\frac{1}{2}(1-\phi_A)+\sqrt{\phi_b\Delta
b+(\frac{1}{2}(1-\phi_A))^2}>0$ where $|\Delta b|<|\Delta A|$, if
$\Delta b\in[0,\Delta b_2]$. \label{p5}
\end{prop}
\begin{proof} Individual supply is $a_i=\hat{P}-\max[\Delta A^2,\Delta b^2]$ under the assumption that $|\Delta b|<|\Delta A|$, we have
 $a_i=\hat{P}-\Delta A^2$. Aggregate supply is thus $A^S=\int_{i\in[0,1]}a_i=\hat{P}-\Delta A^2$. Equilibrium requires $A^S=A^D$ such that
 $A=A_0+\phi_A\Delta A+\phi_b\Delta b-\Delta A^2$, respectively, $\Delta A^2+(1-\phi_A)\Delta A +\phi_b\Delta b=0$.
Solving yields $\Delta
A_{1,2}=-\frac{1}{2}(1-\phi_A)\pm\sqrt{\phi_b\Delta
b+(\frac{1}{2}(1-\phi_A))^2}$. Both roots are real since we
assumed $\phi_b>0$ and $\Delta b>0$. It remains to specify the
conditions under which our initial hypothesis $|\Delta b|<|\Delta
A|$ holds. We start with $\Delta
A_{1}=-\frac{1}{2}(1-\phi_A)-\sqrt{\phi_b\Delta
b+(\frac{1}{2}(1-\phi_A))^2}$ and note (i) $\Delta A_1(\Delta
b=0)=-(1-\phi_A)$ such that $|\Delta b|=0<|\Delta A|$, (ii) the
derivative $\frac{\partial \Delta A}{\partial \Delta
b}=\frac{-\frac{\phi_b}{2}}{\sqrt{\phi_b\Delta
b+(\frac{1}{2}(1-\phi_A))^2}}$ vanishes as $\Delta b$ becomes
large. Taken together, (i) and (ii) imply that an upper bound
$\Delta b_1$ exists, such that $|\Delta b|<|\Delta A|$ as long as
$\Delta b\in[0,\Delta b_1]$. Similarly, regarding the second
equilibrium $\Delta A_2=-\frac{1}{2}(1-\phi_A)+\sqrt{\phi_b\Delta
b+(\frac{1}{2}(1-\phi_A))^2}$, we note that (i) $\Delta A_2(\Delta
b=0)=0$ such that $|\Delta b|=|\Delta A|=0$, (ii) the derivative
$\frac{\partial \Delta A}{\partial \Delta
b}=\frac{\frac{\phi_b}{2}}{\sqrt{\phi_b\Delta
b+(\frac{1}{2}(1-\phi_A))^2}}=_{|\Delta
b=0}\frac{\phi_b}{(1-\phi_A)}$ and $\underset{\Delta
b\rightarrow\infty}{\lim}\frac{\partial \Delta A}{\partial \Delta
b}=\underset{\Delta
b\rightarrow\infty}{\lim}\frac{\frac{\phi_b}{2}}{\sqrt{\phi_b\Delta
b+(\frac{1}{2}(1-\phi_A))^2}}=0$. Taken together, (i) and (ii)
imply that if $\frac{\phi_b}{1-\phi_A}>1$ there exists an upper
bound $\Delta b_2$ such that $|\Delta b|<|\Delta A|$ as long as
$\Delta b\in[0,\Delta b_2]$.
\end{proof}

The first equilibrium in Proposition \ref{p5} corresponds to the
perfect foresight equilibrium, $\Delta A=0$ of Proposition
\ref{p1}: In the limit, where the parameter change becomes
infinitesimally small, we obtain $\Delta A=0$. In this
equilibrium, increases in $b$ indeed increase equilibrium supply
provided that these increases are small such that $\Delta b<\Delta
b_1<\Delta b_2$. In the second equilibrium, which corresponds to
the crisis equilibrium in Proposition \ref{p1}, output is strictly
decreasing in $b$.

Taken together, Propositions \ref{p4} and \ref{p5} suggest that
bold interventions by the government tend to reduce economic
activity as such changes make it more difficult for agents to
forecast prices. Moreover, in the crisis equilibrium $\Delta A_2$,
government interventions, which would increase output were there
no computational frictions, always reduce income.


\section{Learning The REE}\label{ALR}

We have shown that agents can learn the $A_0$ equilibrium if (i)
demand is convex and (ii) agents always coordinate on the ``+"
equilibrium. We now show that agents also learn the REE if (i)
demand is not convex, and (ii) when agents, e.g., alternate
between playing the ``+" and ``-'' root equilibria.

We write the equilibrium condition as:
\begin{eqnarray} A-\phi(A^*)=\phi_A(A^*)(A-A^*)-(A-A^*)^2. \label{ler1}  \end{eqnarray}
The left-hand side of (\ref{ler1}) represents the difference
between supply and demand in an equilibrium where agents use their
knowledge of $A^*,\phi(A^*)$ to estimate the price. This
difference is $0$ in the perfect foresight equilibrium where
$A_0=\phi(A_0)$. Regarding the right-hand side, we define
$\varepsilon=(A-A^*)$, which yields
\begin{eqnarray} A-\phi(A^*)=\phi_A\varepsilon-\varepsilon^2, \label{ler2} \end{eqnarray}
over time as agents observe an increasing number of price quantity
pairs $\{\phi(A_t),A_t\}_{t=0}^{T}$. The distance
$\varepsilon=(A-A^*)$ between the aggregate supply $A$ and the
point of estimation $A^*$ will go to zero.\footnote{To see this
note that, given our assumptions on $\phi()$, all real-valued
equilibrium supplies, (\ref{l244}), fall into a compact interval
$[0,\hat{A}]$. In turn, the sequence of equilibrium supplies
$\{A_t\}_{t=0}^{T}$, for which agents know demand, partitions this
interval. As time progresses, this partition becomes finer and
finer. That is, if we order the quantities $A_t$ such that
$A_l<A_{l+1}, l=1,2,3...T$, the interval between $A_l$ and
$A_{l+1}$ is either filled with new equilibria, or, if there exist
no equilibria between them, economic activity takes place
elsewhere on the interval $[0,\hat{A}]$. Finally, we note that
agents using $A^*$ to think about a deviation from $A^*$, might
choose a quantity $A(A^*)$ which is closer to $A^{**}$ than to
$A^*$. Thus, to forecast the selling price $\phi(A(A^*))$, agents
would rather use $A^{**}$ as the point of approximation. In turn,
once agents use $A^{**}$ they might choose a supply $A(A^{**})$,
which, however, is again closer to $A^{*}$ than $A^{**}$. Thus
agents would switch back to $A^{*}$, and so on. In such
situations, there exists no symmetric equilibrium between points
$A^{*}$ and $A^{**}$. Instead, Appendix \ref{A3} shows that there
does exist an asymmetric equilibrium, where a fraction $\psi$ of
the agents uses point $A^*$ and the remaining fraction $1-\psi$
uses point $A^{**}$, resulting in equilibrium supply
$A=\frac{A^{*}+A^{**}}{2}$.}\label{fn} Hence, from (\ref{ler2}) we
have $\underset{\varepsilon\rightarrow 0}{\lim}(A-\phi(A^*))=0$.
However, the only point where $A=\phi(A)$ is $A_0$, the perfect
foresight equilibrium quantity. Hence, as agents learn more data
on demand, they eventually move
towards the efficient equilibrium. 

\section{Asymmetric Equilibria}\label{A3} Agents using
$\phi(A^*)$ to think about a deviation from $A^*$, might choose a
quantity $A(A^*)$ which is closer to $A^{**}$ than to $A^*$. Thus,
to forecast the selling price $\phi(A(A^*))$, they would rather
use $A^{**}$ as the point of departure. In turn, once agents use
$A^{**}$ they might choose a supply $A(A^{**})$, which, however,
is again closer to $A^{*}$ than $A^{**}$. Thus, given (\ref{l1}),
agents would switch back to $A^{*}$, and so forth:
\begin{eqnarray} &&|A(A^*)-A^*|>|A(A^*)-A^{**}|\label{hh0}\\
&& |A(A^{**})-A^{**}|>|A(A^{**})-A^{*}|.\label{hh1}\end{eqnarray}
In such a situation, there exists no symmetric equilibrium between
points $A^{*}$ and $A^{**}$. Instead, there exists an asymmetric
equilibrium in which a mass $\psi\in(0,1)$ of agents use $A^{*}$
and a mass $1-\psi$ use $A^{**}$.

It follows from (\ref{l1}) that agents are only indifferent
between using $A^*$ and $A^{**}$ if the equilibrium quantity $A$
satisfies $A=\frac{A^{*}+A^{**}}{2}$. That is, agents see $A^*$
and $A^{**}$ as equally informative once the supply they want to
learn about is equally distant from both points. To establish the
existence of an equilibrium we have to prove that there exist
shares $\psi$ and $1-\psi$ for which the equilibrium supply
$A_{\psi}\psi+A_{1-\psi}(1-\psi)$ indeed equals
$\bar{A}=\frac{A^{*}+A^{**}}{2}$.

Once we denote the equilibrium supply by $A(\psi)=\psi
A(A^*,A(\psi))+(1-\psi)A(A^{**},A(\psi))$, we can write the
equations that determine equilibrium as
\begin{eqnarray} &&A(\psi)=\bar{A}\quad \bar{A}:=\frac{A^{*}+A^{**}}{2}\label{h1}\\
&& A(\psi):=\psi
A_{\psi}(A^*,A(\psi))+(1-\psi)A_{1-\psi}(A^{**},A(\psi))\label{h11}\\
&& A_{\psi}(A^*,A(\psi))=\hat{P}_{\psi}-(A(\psi)-A^*)\quad
A_{\psi}(A^{**},A(\psi))=\hat{P}_{1-\psi}-(A(\psi)-A^{**})\\
&&\hat{P}_{\psi}=\phi(A^*)+\phi_A(A^*)(A(\psi)-A^*)\\
&&\hat{P}_{1-\psi}=\phi(A^{**})+\phi_A(A^{**})(A(\psi)-A^{**})\label{h2}
\end{eqnarray}
Solving (\ref{h11})-(\ref{h2}) for supply $A(\psi)$ yields:
\begin{eqnarray} \hspace{-1cm}&& A(\psi)=-\frac{p}{2}\pm\sqrt{-q+(\frac{p}{2})^2}\nonumber\\
\hspace{-1cm}&&
-q=\psi\phi(A^*)+(1-\psi)\phi(A^{**})-\psi\phi_A(A^*)-(1-\psi)\phi_AA^{**}-\psi
(A^*)^2-(1-\psi)(A^{**})^2\nonumber\\
\hspace{-1cm}&&p=1+2\psi
A^*+2(1-\psi)A^{**}-\psi\phi_A(A^*)-(1-\psi)\phi_A(A^{**})>0\nonumber
\end{eqnarray}
Given our assumption that supply is downward-sloping, we have
$-p<0$. Hence, there is only one real root
$A(\psi)=-\frac{p}{2}+\sqrt{-q+(\frac{p}{2})^2}$ with positive
supply. Note in particular that $A(\psi=0)$ ($A(\psi=1)$) is the
supply when all agents use $A^*$ ($A^{**}$) as the point of
reasoning. By our initial hypothesis (\ref{hh0})-(\ref{hh1}), we
have $A(\psi=0)>\bar{A}$ and $A(\psi=1)<\bar{A}$. Hence, there
exists an intermediate value $\tilde{\psi}$, at which
$A(\tilde{\psi})=\bar{A}$, if $\sqrt{-q+(\frac{p}{2})^2}$ as
required by the indifference condition (\ref{h1}).\footnote{The
root $\sqrt{-q+(\frac{p}{2})^2}$ is never complex as $\psi$ runs
from $0$ to $1$. To see this we recall $-q(\psi=0)>0$,
$-q(\psi=1)>0$ and that the derivative $\frac{\partial
-q}{\partial \psi}$ is not a function of $\psi$ itself and thus
cannot change signs as $\psi$ varies. Accordingly, $-q>0 \forall
\psi\in[0,1]$. Hence, $\sqrt{-q+(\frac{p}{2})^2}$ is always real.}

\section{Asymmetric Information Equilibria}\label{A2.1}

We solve the model in two steps. First, we guess the equilibrium
outcome. Given this guess, we solve the estimation problem for
agent $i$. Second, we solve for the equilibrium and verify our
guess.

\paragraph{Problem of an individual agent:} To find the optimal supply of an individual agent,
we must solve (\ref{ha2})-(\ref{ha3}). To do so, we start with the
\emph{guess} that agent $i$ beliefs that agents $j$ hold the same
believe over aggregate supply that he holds himself, i.e.,
$\hat{A}|A_i=(\hat{A}|A_j)|A_i=\hat{A}$. Given this guess, we show
that:
\begin{eqnarray} \hat{P}|A_i=(\hat{P}|A_j)|A_i. \label{ha4} \end{eqnarray}
Put differently, equation (\ref{ha4}) means that agent $i$
believes that agents $j$ observe points, which lie on his
estimated demand curve. That is, agent $i$ knows that agent $j$
estimates the price as
\begin{eqnarray} \hat{P}|A_j=\phi(A_j)+\phi_A(A_j)(\hat{A}|A_j-A_j), \label{ha5}\nonumber \end{eqnarray}
and thus $i$ estimates $\hat{P}|A_j$ as:
\begin{eqnarray} (\hat{P}|A_j)|A_i=\phi(A_j)|A_i+\phi_A(A_j)|A_i((\hat{A}|A_j)|A_i-A_j), \label{ha6} \end{eqnarray}
where
\begin{eqnarray}&& \phi(A_j)|A_i=\phi(A_i)+\phi_A(A_i)(A_j-A_i),\label{ha7}\\
&&\phi_A(A_j)|A_i=\phi_A(A_i),\quad
\hat{A}|A_i=(\hat{A}|A_j)|A_i=\hat{A}. \label{ha8}
\end{eqnarray}
Taken together, (\ref{ha6})-(\ref{ha8}) mean that agent $i$ uses
polynomials to approximate the demand curve upon which the other
player observes a point $A_j,\phi(A_j),\phi_A(A_j)$. Using
(\ref{ha7})-(\ref{ha8}), (\ref{ha6}) rewrites:
\begin{eqnarray} (\hat{P}|A_j)|A_i&&=\phi(A_i)+\phi_A(A_i)(A_j-A_i)+\phi_A(A_i)(\hat{A}-A_j),\nonumber\\
                                    &&=\phi(A_i)+\phi_A(A_i)(\hat{A}|A_i-A_i)=_{|(\ref{ha2})}\hat{P}|A_i. \label{ha9}\nonumber \end{eqnarray}
Agent $i$ thus believes that agent $j$ will work with a price
estimate that is identical to the one he uses himself and thus he
will conclude that $j's$ supply forecast is identical to his own,
such that $\hat{A}|A_i=(\hat{A}|A_j)|A_i=\hat{A}$, which confirms
our initial guess. It remains to solve (\ref{ha2})-(\ref{ha3}) for
player $i's$ forecast in the two equilibria $A_{1,2}$:
\begin{eqnarray} \hat{A}_{1,2}|A_i&&=A_i-\frac{1-\phi_A(A_i)}{2}\pm\sqrt{(\phi(A_i)-A_i)+\Big(\frac{1-\phi_A(A_i)}{2}\Big)^2},\label{eq1}\nonumber\\
a_i&&=\hat{P}|A_i-(\hat{A}|A_i-A_i)^2\label{eq11}\nonumber\\
&&=\phi(A_i)+\phi_A(A_i)(\hat{A}|A_i-A_i)-(\hat{A}|A_i-A_i)^2\label{eq111}\\
a_{i;1,2}&&=\phi(A_i)+\phi_A(A_i)\Big(-\frac{1-\phi_A(A_i)}{2}\pm\sqrt{(\phi(A_i)-A_i)+\Big(\frac{1-\phi_A(A_i)}{2}\Big)^2}\Big)\nonumber\\
&&-\Big(-\frac{1-\phi_A(A_i)}{2}\pm\sqrt{(\phi(A_i)-A_i)+\Big(\frac{1-\phi_A(A_i)}{2}\Big)^2}\Big)^2.\label{eq21}\label{eq2}
\end{eqnarray}


\paragraph{Equilibrium:}
From (\ref{eq2}) we calculate equilibrium supply as:
\begin{eqnarray} A_{1,2}=\int_{[0,1]}a_{1,2}(A_i)f(A_i)di \label{ha1.1}. \end{eqnarray}
Concerning the equilibrium quantity $A_k, k=1,2$, we note

\begin{lemma}\label{Lem4} If demand $\phi$ is quasi-convex, then equilibrium
output across both equilibria $A_k, k=1,2$ falls short of
efficient output $A_0$.\end{lemma}
\begin{proof} To compare equilibrium output (\ref{ha1.1}) to efficient output $A_0$, we recall that $A_0=\phi(A_0)$.
We also recall that, for convex functions $f$, we have
$f'(u)\leq\frac{f(v)-f(u)}{v-u}$. In the current context, this
means that agents tend to underestimate convex demand functions:
\begin{eqnarray} \phi(A_i)+\phi_A(A_i)(A-A_i)\leq\phi(A)\nonumber\end{eqnarray}
recalling (\ref{eq111}) we have:
\begin{eqnarray}a_i=\phi(A_i)+\phi_A(A_i)(\hat{A}|A_i-A_i)-(\hat{A}|A_i-A_i)^2\leq\phi(A_0)-(\hat{A}|A_i-A_i)^2\leq\phi(A_0)=A_0. \nonumber\end{eqnarray}
Put differently, agents supply less than the efficient quantity
since (i) they underestimate demand and (ii) since they know that
their price estimate is inaccurate.
\end{proof}

The equilibria in (\ref{ha1.1}), feature two sources of
inefficiency. First, aggregate output falls short of the efficient
level $A_0$. Second, since price estimates vary, output and the
marginal cost of output differ across firms. Aggregate output is
thus produced inefficiently.

\end{appendix}


\newpage

\newpage

\begin{appendix}

\addcontentsline{toc}{section}{References}
\markboth{References}{References}
\bibliographystyle{apalike}
\bibliography{References}

\end{appendix}

\end{document}